\documentclass [12pt,a4paper]{amsart}

\usepackage{amssymb} 

\usepackage{graphicx}

\usepackage{graphics}
\usepackage{amsmath}
\usepackage{amsthm}

\usepackage{enumerate}

\newtheorem{thm}{Theorem}

\newtheorem{cor}[thm]{Corollary}

\newtheorem{prop}[thm]{Proposition}

\theoremstyle{definition}

\newtheorem*{remark}{Remark}

\def\qed{\nobreak\hfill $\square$}

\def\<{\langle}

\def\>{\rangle}

\def\det{\mathrm{Det}}

\def\bM{{\bf M}}

\def\bM{{\mathbf{M}}}

\def\M3{M_3(\bbbc)}

\def\M3r{M_3(\bbbr)}

\def\diag{\mathrm{Diag}}

\def\bbbr{{\mathbb R}}

\def\bbbc{{\mathbb C}}

\newcommand{\R}{\mathbb{R}}

\newcommand{\C}{\mathbb{C}}

\newcommand*{\be}{\begin{equation}}

\newcommand*{\ee}{\end{equation}}

\newcommand*{\bit}{\begin{itemize}}
\newcommand*{\eit}{\end{itemize}}
\newcommand*{\ben}{\begin{enumerate}}
\newcommand*{\een}{\end{enumerate}}

\newcommand{\tr}{\mathrm{Tr}}

\newcommand{\norm}[1]{\left\|#1\right\|}

\newcommand*{\inner}[2]{\left<#1,\,#2\right>}

\newcommand*{\ler}[1]{\left(#1\right)}

\newcommand{\ba}{\begin{array}}
\newcommand{\ea}{\end{array}}

\newcommand{\bu}{\mathbf{u}}
\newcommand{\bv}{\mathbf{v}}
\newcommand{\bB}{\mathbf{B}}

\newcommand{\fP}{\mathbb{P}}
\newcommand{\fD}{\mathbb{D}}

\begin{document}
\title{On algebraic endomorphisms of the Einstein gyrogroup}

\author{LAJOS MOLN\'AR}
\address{MTA-DE "Lend\" ulet" Functional Analysis Research Group, Institute of Mathematics\\
         University of Debrecen\\
         H-4010 Debrecen, P.O. Box 12, Hungary}
\email{molnarl@science.unideb.hu}
\urladdr{http://www.math.unideb.hu/\~{}molnarl/}

\author{D\'ANIEL VIROSZTEK}
\address{Institute of Mathematics\\ 
         Budapest University of Technology and Economics\\
         H-1521 Budapest, Hungary} 
\email{virosz@math.bme.hu} 
\urladdr{http://www.math.bme.hu/~virosz}
\thanks{The first author was supported by the "Lend\" ulet" Program (LP2012-46/2012) of the Hungarian Academy of Sciences. The second author was partially supported by the Hungarian Scientific Research Fund (OTKA) Reg. No.  K104206}

\begin{abstract}
We describe the structure of all continuous algebraic endomorphisms of the open unit ball $\bB$ of $\R^3$ equipped with the Einstein velocity addition. We show that any nonzero such transformation originates from an orthogonal linear transformation on $\R^3$.
\end{abstract}
\maketitle
\section{Introduction}

Velocity addition was defined by Einstein in his famous paper of 1905 which founded the special theory of relativity. In fact, the whole theory is essentially based on Einstein velocity addition law, see \cite{ein}. The algebraic structure corresponding to this operation is a particular example of so-called gyrogroups the general theory of which has been developed by Ungar \cite{ung}. 

\par
The Einstein gyrogroup of dimension three is the pair $(\bB, \oplus),$ where $\bB=\{\bu\in \R^3 : \norm{\bu}< 1 \}$ and $\oplus$ is the binary operation on $\bB$ given by
\be \label{relsum}
\oplus: \bB \times \bB \rightarrow \bB; \, (\bu,\bv)\mapsto \bu \oplus \bv:= \frac{1}{1+\inner{\bu}{\bv}}\ler{\bu+\frac{1}{\gamma_\bu}\bv+\frac{\gamma_\bu}{1+\gamma_\bu}\inner{\bu}{\bv} \bu},
\ee
where $\gamma_\bu=\ler{1-\norm{\bu}^2}^{-\frac{1}{2}}$ is the so-called Lorentz factor.
The operation $\oplus$ is called Einstein velocity addition or relativistic sum (cf. \cite{abe,kim}).
Here and throughout this paper, $\inner{\cdot}{\cdot}$ stands for the usual Euclidean inner product and $\norm{\cdot}$ denotes the induced norm.
\par
The study of automorphisms (more generally, endomorphisms) of algebraic structures is of special importance in most areas of both mathematics and mathematical physics.
The aim of this note is to determine the (continuous) endomorphisms (in particular, automorphisms) of the fundamental structure $(\bB, \oplus)$ of special relativity theory. 
Important information on isomorphisms, automorphisms (symmetries) of quantum structures can be found in \cite{CasVitLahLev04} and other sorts of so-called preservers on similar structures are discussed in \cite{MB}, Chapter 2.

The main theorem of this paper is obtained as an application of our recent result on so-called Jordan triple endomorphisms of $2 \times 2$ positive definite matrices \cite[Theorem 1]{lmdv}.
The other ingredient of our argument is the result \cite[Theorem 3.4]{kim} of Kim. The discussion below may look rather simple but the mathematical facts and results that we combine are highly nontrivial. 

Our main result reads as follows.

\begin{thm} \label{tmain}
Let $\beta: \bB \rightarrow \bB$ be a continuous map. We have $\beta$ is an algebraic endomorphism with respect to the operation $\oplus$, i.e., $\beta$ satisfies
$$
\beta (\bu \oplus \bv)=\beta (\bu) \oplus \beta(\bv), \quad \bu,\bv\in \bB
$$
if and only if
\begin{itemize}
\item[(i)]
either there is an orthogonal matrix $O \in \bM_3(\R)$ such that $$\beta(\bv)=O\bv, \quad \bv\in \bB;$$
\item[(ii)]
or we have
$$
\beta(\bv)=0, \quad \bv\in \bB.
$$
\end{itemize}
\end{thm}

Here continuity refers to the usual topology on $\bB$ inherited from the Euclidean space $\R^3$. For a related comment see the remarks at the end of the paper.

By the above result we have the interesting conclusion that the group of all (continuous) automorphisms of the Einstein gyrogroup $(\bB,\oplus)$ coincides with the orthogonal group of $\R^3$.

\section{Open Bloch ball, qubit density matrices, and $2 \times 2$ positive definite matrices of determinant one}

To prove our main result we need 
an important observation made by Kim what we present below.

We denote by $\fP_2$ the set of all $2 \times 2$ positive definite complex matrices.
Let $\fD$ stand for the set of all $2\times 2$ regular density matrices, i.e., the collection of all  elements of $\fP_2$ with trace 1, $$\fD=\{A \in \fP_2\, | \, \tr A=1\}.$$ From the quantum theoretical point of view, $\fD$ is the set of all regular density matrices of the $2$-level quantum system.
One can define a binary operation $\odot$ on $\fD$ as
\begin{equation*} 
\odot: \fD \times \fD \rightarrow \fD; \, (A,B) \mapsto A \odot B:= \frac{1}{\tr AB} A^{\frac{1}{2}} B A^{\frac{1}{2}}.
\end{equation*}
For certain reasons, we call $\odot$ the normalized sequential product.
\par
The well-known Bloch parametrization of regular density matrices is the following map:
\begin{equation*} 
\rho: \R^3 \supset \bB \rightarrow \bM_2(\C); \, \left[\ba{c}v_1\\v_2\\v_3\ea\right]=\bv \mapsto \rho(\bv):=\frac{1}{2} \left[\ba{cc}1+v_3 & v_1- i v_2 \\ v_1 + i v_2 & 1-v_3 \ea\right].
\end{equation*}
The transformation $\rho$ is clearly a bijection between $\bB$ and $\fD,$ and in fact, by \cite[Theorem 3.4]{kim}, much more is true.
\begin{thm}[S. Kim] \label{T:kim}
The Bloch parametrization $\rho: (\bB, \oplus) \rightarrow (\fD, \odot); \, \bv \mapsto \rho(\bv)$ is an isomorphism.
\end{thm}

Throughout this note the word 'isomorphism' refers to a bijective map between algebraic structures which respects (preserves) the relevant algebraic operation(s).
\par
Let us now consider a structure which is similar to the space of $2 \times 2$ regular density matrices equipped with the normalized sequential product. Namely, Let $\fP_2^1$ be the set of all $2 \times 2$ positive definite matrices with determinant $1.$ The sequential product $\boxdot$ on $\fP_2^1$ is defined as
\begin{equation*} 
\boxdot: \fP_2^1 \times \fP_2^1 \rightarrow \fP_2^1; \, (A,B) \mapsto A \boxdot B:= A^{\frac{1}{2}} B A^{\frac{1}{2}}.
\end{equation*}
We show that $(\fD, \odot)$ and $(\fP_2^1, \boxdot)$ are isomorphic structures.
\begin{prop}\label{cl1}
The map $\tau: (\fD, \odot) \rightarrow (\fP_2^1, \boxdot); \, A \mapsto \tau(A):=\frac{1}{\sqrt{\det A}} A$ is an isomorphism.
\end{prop}

\begin{proof}
To prove the injectivity assume that $\tau(A)=\tau(B)$ for some $A,B \in \fD.$ That is, $\frac{1}{\sqrt{\det A}} A=\frac{1}{\sqrt{\det B}} B,$ which means that $A$ is a positive scalar multiple of $B.$ By $\tr A=\tr B=1$ we can deduce that $\sqrt{\det A}=\sqrt{\det B}$ and therefore $A=B.$

For any $A \in \fP_2^1$ we have $\frac{1}{\tr A} A \in \fD.$ By $\det A=1$ it follows that
$$\tau\ler{\frac{1}{\tr A} A}=\frac{1}{\sqrt{\det \ler{\frac{1}{\tr A} A}}}\frac{1}{\tr A} A=\frac{1}{\frac{\sqrt{\det A}}{\tr A}}\frac{1}{\tr A} A=A.$$
This shows the surjectivity of $\tau$.

Finally, we need to show that $\tau$ respects the operations $\odot,\boxdot$. 
Using the properties of the determinant, for any $A,B\in \fD$ we compute
$$
\tau\ler{A \odot B}= \frac{1}{\sqrt{\det \ler{\frac{1}{\tr AB} A^{\frac{1}{2}} B A^{\frac{1}{2}}}}}\frac{A^{\frac{1}{2}} B A^{\frac{1}{2}}}{\tr AB}=\frac{\tr AB}{\sqrt{\det \ler{A^{\frac{1}{2}} B A^{\frac{1}{2}}}}}\frac{A^{\frac{1}{2}} B A^{\frac{1}{2}}}{\tr AB}
$$
$$
=\ler{\frac{A}{\sqrt{\det A}}}^{\frac{1}{2}}\frac{B}{\sqrt{\det B}} \ler{\frac{A}{\sqrt{\det A}}}^{\frac{1}{2}}=\frac{A}{\sqrt{\det A}} \boxdot \frac{B}{\sqrt{\det B}}=\tau(A)\boxdot\tau(B).
$$
This completes the proof.
\qed
\end{proof}

Observe that the inverse of $\tau$ is given by $\tau^{-1}(A)=\frac{1}{\tr A} A$, $A\in \fP_2$.

\section{Proof of the main result}

To verify the main result of the paper let us recall the following recent result of ours \cite[Theorem 1]{lmdv} which, beside Kim's observation, is the second main ingredient of the proof of Theorem~\ref{tmain}. It may look surprising but its content, i.e., the description of the structure of the continuous so-called Jordan triple endomorphisms of $\fP_2$, was an open problem for quite a while and the solution we have finally found is rather complicated resting on highly nontrivial arguments and facts.
 
Below continuity of maps on matrix structures refers to any one of the equivalent linear norm topologies on the full matrix algebra.    

\begin{thm}\label{jtrip}
Let $\phi: \fP_2 \rightarrow \fP_2$ be a continuous map.
Assume that it is Jordan triple endomorphism, i.e., $\phi$ satisfies
$$
\phi(ABA)=\phi(A)\phi(B)\phi(A),\quad A,B\in \fP_2.
$$
Then $\phi$ is of one of the following forms:
\begin{itemize}
\item[(1)]
there is a unitary matrix $U\in \bM_2(\C)$ and a real number $c$ such that $$\phi(A)=(\det A)^c UAU^*, \quad A\in \fP_2;$$
\item[(2)]
there is a unitary matrix $V\in \bM_2(\C)$ and a real number $d$ such that $$\phi(A)=(\det A)^d VA^{-1}V^*, \quad A\in \fP_2;$$
\item[(3)]
there is a unitary matrix $W\in \bM_2(\C)$ and real numbers $c_1,c_2$ such that 
$$
\phi(A)=W\diag [(\det A)^{c_1}, (\det A)^{c_2}]W^*, \quad A\in \fP_2.
$$
\end{itemize}
\end{thm}

Using this theorem the continuous sequential endomorphisms of $\fP_2^1$ can be described as follows.
\begin{cor}\label{p21}
Let $\phi: \fP_2^1 \rightarrow \fP_2^1$ be a continuous endomorphism with respect to the operation $\boxdot$ meaning that $\phi$ satisfies 
$$
\phi(A\boxdot B)=\phi(A)\boxdot \phi(B),\quad A,B\in \fP_2^1.
$$
Then $\phi$ is of one of the following forms:
\begin{itemize}
\item[(1)]
there is a unitary matrix $U\in \bM_2(\C)$ such that $$\phi(A)=UAU^*, \quad A\in \fP_2^1;$$
\item[(2)]
there is a unitary matrix $V\in \bM_2(\C)$ such that $$\phi(A)=VA^{-1}V^*, \quad A\in \fP_2^1;$$
\item[(3)] we have
$$
\phi(A)=I, \quad A\in \fP_2^1.
$$
\end{itemize}
\end{cor}

\begin{proof}
If $\phi: \fP_2^1 \rightarrow \fP_2^1$ is a sequential endomorphism, then it is a Jordan triple endomorphism, as well. Indeed, $\phi(A^2)=\phi(A\boxdot A)=\phi(A)\boxdot \phi(A)=\phi(A)^2$ holds for all $A \in \fP_2^1.$ It follows that $\phi(ABA)=\phi(A^2 \boxdot B)=\phi(A^2)\boxdot \phi(B)=\ler{\phi(A)^2}^{\frac{1}{2}}\phi(B)\ler{\phi(A)^2}^{\frac{1}{2}}=\phi(A)\phi(B)\phi(A)$ for all $A,B \in \fP_2^1.$
\par
The map
$$
\psi: \fP_2 \rightarrow \fP_2; \, A\mapsto\psi(A):=\sqrt{\det A}\cdot\phi\ler{\frac{A}{\sqrt{\det A}}}
$$
is clearly a continuous Jordan triple endomorphism of $\fP_2$  which extends $\phi$ (the idea of the definition of $\psi$ comes from \cite[proof of Theorem 3]{ML15b}). Now, the statement is an immediate consequence of the previous theorem.
\qed 
\end{proof}

Using the isomorphism $\tau$ defined in Proposition~\ref{cl1} which is clearly a homeomorphism, too, we can pull back the structural result on the continuous endomorphisms of $(\fP_2^1, \boxdot)$ to $(\fD, \odot).$ Namely, the continuous endomorphism of $(\fD, \odot)$ are exactly the maps of the form
$$\tau^{-1} \circ \phi \circ \tau,$$
where $\phi$ is a continuous endomorphism of $(\fP_2^1, \boxdot).$
The following corollary can be verified by straightforward computations.
\begin{cor} \label{qubit}
Let $\alpha: \fD \rightarrow \fD$ be a continuous endomorphism with respect to the operation $\odot$. Then $\alpha$ is of one of the following forms:
\begin{itemize}
\item[(1)]
there is a unitary matrix $U\in \bM_2(\C)$ such that $$\alpha(A)=UAU^*, \quad A\in \fD;$$
\item[(2)]
there is a unitary matrix $V\in \bM_2(\C)$ such that $$\alpha(A)=\frac{VA^{-1}V^*}{\tr A^{-1}}, \quad A\in \fD;$$
\item[(3)] we have
$$
\alpha(A)=I/2, \quad A\in \fD.
$$
\end{itemize}
\end{cor}

Putting all information we have together, the proof of the main result is now easy.

\begin{proofs}
We have learned from the result Theorem~\ref{T:kim} due to Kim that the Bloch parametrization $\rho$ is an isomorphism between $(\bB, \oplus)$ and $(\fD, \odot)$. Clearly, $\rho$ is a homeomorphism, too. Therefore, the continuous endomorphisms of $(\bB, \oplus)$ are exactly the maps of the form $\beta=\rho^{-1} \circ \alpha \circ \rho,$ where $\alpha$ is a continuous endomorphism of $(\fD,\odot).$
\par
By Corollary~\ref{qubit} there are three possibilities.
Assume first that we have a unitary $U\in \bM_2(\C)$ such that 
$\alpha(A)=UAU^*,$ $A\in \fD$. Denote by ${\bf H}_2^0(\C)$ the linear space of all traceless self-adjoint $2\times 2$ complex matrices and equip this space with the inner product $\langle A,B\rangle := \frac{1}{2} \tr AB$, $A,B\in {\bf H}_2^0(\C)$. Define
\begin{equation*} 
\gamma: \R^3 \rightarrow {\bf H}_2^0(\C); \, \left[\ba{c}v_1\\v_2\\v_3\ea\right]=\bv \mapsto \gamma(\bv):=\left[\ba{cc}v_3 & v_1- i v_2 \\ v_1 + i v_2 & -v_3 \ea\right].
\end{equation*}
Clearly, $\gamma$ is a linear isomorphism from $\R^3$ onto ${\bf H}_2^0(\C)$ which preserves the inner product. Define $\tilde \alpha : {\bf H}_2^0(\C) \to {\bf H}_2^0(\C)$ by $\tilde \alpha(A)=UAU^*$, $A\in {\bf H}_2^0(\C)$. Then $O:=\gamma^{-1} \circ \tilde \alpha \circ \gamma$ is an orthogonal linear transformation on $\R^3$ and using the relation $\gamma (\bv)=2\rho(\bv)-I$, $\bv\in \bB$, we easily deduce that
$\rho \circ \alpha \circ \rho=O$ holds.
\par
If $\alpha(A)=\frac{VA^{-1}V^*}{\tr A^{-1}},$ $A\in \fD$, then
the conclusion follows from the previous case. The only thing we have to observe is that
$\frac{(\rho(\bv))^{-1}}{\tr (\rho(\bv))^{-1}}= \rho(-\bv)$ holds which follows from \cite[Remark 3.5]{kim}.
\par 
Finally, if $\alpha(A)=I/2$, then we clearly have $\rho^{-1} \circ \alpha \circ \rho=0.$
\par
The converse statement that the formulas in (i) and (ii) 
define continuous endomorphisms of the Einstein gyrogroup is just obvious.
\qed
\end{proofs}

\begin{remark}
We conclude our note with some remarks.

Above we have given the complete description of all continuous endomorphisms of $\bB$ under the operation of Einstein velocity addition.
In the recent paper \cite{abe} Abe have described the automorphisms of the Einstein gyrovector space 
\cite[Theorem 3.1]{abe}. He concludes that those automorphisms are exactly the restrictions of orthogonal linear transformations onto the open unit ball. To see clearly the content of his result which looks very closely related to ours, one needs to be cautious and look at the definition \cite[Definition 2.13]{abe} of automorphisms of gyrovector spaces. In fact, that definition includes the assumption about the preservation of the inner product. Hence, Abe's result says that every bijective map of $\bB$ which preserves the Einstein addition (plus a sort of scalar multiplication) and also preserves the inner product necessarily originates from an orthogonal linear transformation on $\R^3$.

We need to point out that the requirement concerning the inner product preserving property is very strong, it alone implies the above conclusion. Indeed, any inner product preserving map $\phi:\bB \to \bB$ easily extends to an inner product preserving map on $\R^3$, see the proof of \cite[Lemma 4.4]{abe}. Moreover, it is well-known that on any inner product space the inner product preserving maps are automatically linear. That means that Abe's result carries no information concerning the automorphism of $\bB$ endowed merely with the Einstein addition $\oplus$ and the usual topology. Let us remark at this point that it is not difficult to see that the usual topology coincides with the topology generated by the Einstein gyrometric, see \cite[Example 2.10]{abe}.

So we can tell that our result is much different and in fact much stronger than Abe's but we also have to point out that
our result is proved only in three dimension. 
The main reason for this is that Kim's isomorphic identification between the open unit ball and the set of all regular density matrices is valid only in that low dimensional case. Though, for its physical content, the most important case is certainly this one,
it would very be interesting to know what happens in higher dimensions. We propose this as an open problem.
\end{remark}

\end{document}